\documentclass[11pt,a4paper]{article}

\usepackage[margin=1in]{geometry}
\usepackage{amsmath,amssymb,amsfonts,amsthm}
\usepackage{mathtools}
\usepackage{hyperref}
\usepackage{enumitem}
\usepackage{booktabs}
\usepackage{xcolor}
\usepackage{tikz}
\usepackage{pgfplots}
\usepackage{float}
\usepackage{caption}
\usepackage{setspace}
\usepackage[round]{natbib}

\pgfplotsset{compat=1.18}
\usetikzlibrary{arrows.meta, positioning, calc, patterns}
\theoremstyle{plain}
\newtheorem{theorem}{Theorem}[section]

\newtheorem{proposition}[theorem]{Proposition}
\newtheorem{corollary}[theorem]{Corollary}

\theoremstyle{definition}
\newtheorem{definition}[theorem]{Definition}

\newtheorem{example}[theorem]{Example}

\theoremstyle{remark}
\newtheorem{remark}[theorem]{Remark}

\newcommand{\E}{\mathbb{E}}
\newcommand{\Prob}{\mathbb{P}}

\newcommand{\ind}{\mathbf{1}}
\DeclareMathOperator{\PPV}{PPV}
\DeclareMathOperator{\FPR}{FPR}

\DeclareMathOperator*{\esssup}{ess\,sup}

\newcommand{\Lreq}{\Lambda_{\mathrm{req}}}
\newcommand{\Lavail}{\Lambda_{\mathrm{avail}}}
\newcommand{\Lach}{\Lambda_{\mathrm{ach}}}
\newcommand{\tgt}{\tau}

\begin{document}

\title{When AAA Satisfies Nothing\\[0.3em]
Impossibility Theorems for Structured Credit Ratings}

\author{
Marco Pollanen\\
Department of Mathematics and Statistics\\
Trent University\\
Peterborough, ON, Canada\\
\texttt{marcopollanen@trentu.ca}
}

\date{}

\maketitle

\begin{abstract}
\noindent A credit rating of AAA asserts near-certainty of repayment. This paper asks whether the pre-crisis information environment could have supported that assertion for structured products. Bayes' theorem implies that any reliability target requires a minimum level of statistical discrimination between instruments that will repay and those that will not. At structured-finance base rates, a four-nines reliability target demands discrimination on the order of $10{,}000$ to 1. A three-nines target demands $1{,}000$ to 1. Nothing in the published credit-prediction literature provides an affirmative basis for believing that discrimination of this magnitude was achievable with the data available at rating time. Retrospectively, the realized system fell short of the four-nines benchmark by roughly $90{,}000$-fold. The framework accommodates the historical feasibility of corporate AAA ratings, where high base rates and rich information produce low required discrimination. Illustrative calibrations for contemporary collateralized loan obligations suggest that material tension between the precision target and the information environment persists. The central implication is that the AAA precision claim itself likely exceeded what the available information could support.
\end{abstract}

\vspace{0.5em}

\newpage

\section{Introduction}
\label{sec:intro}

On August 9, 2007, BNP Paribas suspended redemptions from three investment funds, citing the ``complete evaporation of liquidity'' in segments of the U.S.\ securitization market. The underlying assets had carried high ratings months earlier. Over the subsequent crisis, vast quantities of AAA-rated structured products were downgraded to speculative grade, and losses reached into the trillions \citep{fcic2011}.

The AAA designation has historically implied near-certainty of repayment. Moody's cumulative default rate data show that AAA-rated corporate bonds default at very low rates over five-year horizons, consistent with reliability on the order of $\tgt \approx 0.999$ \citep{moodys2006}. Yet the same label failed comprehensively for structured products while continuing to perform for corporate bonds.

A substantial literature attributes the failure to misaligned incentives \citep{bolton2012}, model error in the Gaussian copula framework \citep{li2000, donnelly2010}, regulatory arbitrage \citep{acharya2013}, weakened originator screening \citep{keys2010}, and rating shopping \citep{skreta2009, griffin2012}. Most of these explanations treat the failure as \emph{contingent}: better incentives, models, or oversight could have restored the promised precision.

This paper investigates a more fundamental explanation. The question is not why the models were wrong, but whether the precision target was attainable at all given the information available at rating time.

The analysis proceeds in three steps. First, a formal impossibility theorem establishes that any reliability target imposes, through Bayes' theorem, a minimum level of statistical discrimination between instruments that will succeed and those that will fail. An upper bound on achievable discrimination follows from the overlap between the conditional distributions of the rating-time signal. When required discrimination exceeds this bound, no admissible decision rule can achieve the target.

Second, the theorem is calibrated to pre-crisis CDOs. At base rates near $\pi = 0.5$, a four-nines target requires discrimination on the order of $10^4$; a three-nines target requires $10^3$. The paper does not directly estimate the discrimination ceiling but makes a narrower claim: the published evidence provides no affirmative basis for believing that tail discrimination approached the required order of magnitude, and even generous benchmark ceilings imply that the ratio of required to available discrimination (the \emph{tension ratio}, defined formally in Section~\ref{sec:methodology}) exceeds 100. Retrospectively, achieved discrimination fell short of the four-nines requirement by roughly 90{,}000-fold.

Third, the framework is extended to contemporary collateralized loan obligations. Despite a zero-default historical record, the calibration suggests that the informational burden remains severe. A zero-default record from a finite sample dominated by benign conditions does not demonstrate discrimination at the required scale.

Four structural features make the gap resistant to amelioration. The ceiling reflects the information content of the data; no algorithm can extract what is absent. Additional data drawn from a single regime need not identify stress-regime parameters. Weak identification in any parameter block can bottleneck overall discrimination. Deterministic structural engineering (pooling, tranching, waterfall design) operates on the same information and cannot create discrimination the data lack.

The paper contributes to three literatures. First, on the failure of structured ratings: existing work documents catastrophic sensitivity to correlation assumptions \citep{coval2009}, weak explanatory power of observable CDO characteristics \citep{benmelech2009}, and discretionary rating inflation \citep{griffin2012}. This paper adds that even absent these failures, the precision target exceeded what the information environment could support. Second, on rating agency incentives: \citet{bolton2012} and \citet{skreta2009} explain why ratings were \emph{inflated}; this paper asks whether the target was achievable even without incentive distortions. Third, on cross-asset rating comparability: \citet{cornaggia2017} document persistent performance differences between corporate and structured ratings. The tension ratio provides a decision-theoretic explanation for why the information required to sustain a given precision level differs sharply across asset classes.

Section~\ref{sec:methodology} develops the framework. Section~\ref{sec:empirical} calibrates to pre-crisis CDOs. Section~\ref{sec:results} validates against corporate bonds and extends to contemporary CLOs. Section~\ref{sec:discussion} discusses implications. Section~\ref{sec:conclusion} concludes.

\section{Framework and Impossibility Theorem}
\label{sec:methodology}

\subsection{The classification problem}

At the AAA threshold, the rating problem reduces to binary classification. Rating agencies sometimes characterize ratings as ordinal rank-orderings rather than cardinal probability statements. At the AAA level, however, the label carries a specific near-zero-default promise to investors and regulators, functioning as a cardinal certification of tail safety.

\begin{definition}[State space]
\label{def:state}
Let $E \in \{0,1\}$ denote the binary event the rating certifies over a fixed horizon, where $E=1$ denotes success (no credit event) and $E=0$ denotes failure. The \emph{base rate} $\pi := \Prob(E = 1 \mid C)$ is the proportion of genuinely successful instruments within a specified reference class $C$ of candidates.
\end{definition}

\begin{definition}[Information and decision rules]
\label{def:decision}
Let $X \in \mathcal{X}$ denote the totality of information available at rating time. A \emph{decision rule} is a measurable function $\phi: \mathcal{X} \to \{0,1\}$, where $\phi(x) = 1$ assigns a AAA rating.
\end{definition}

\begin{definition}[Performance metrics]
\label{def:metrics}
For a decision rule $\phi$, the \emph{sensitivity} is $S(\phi) := \Prob(\phi(X) = 1 \mid E = 1)$, the \emph{false positive rate} is $F(\phi) := \Prob(\phi(X) = 1 \mid E = 0)$, and the \emph{positive predictive value} is $\PPV(\phi) := \Prob(E = 1 \mid \phi(X) = 1)$.
\end{definition}

All probabilities are conditional on the reference class $C$, suppressed notationally. The PPV is the quantity that matters for investors: a AAA rating claiming 99.99\% reliability asserts $\PPV \geq 0.9999$.

Throughout, we restrict to $\pi \in (0,1)$ and $\tgt \in (0,1)$. An \emph{admissible} decision rule satisfies $\Prob(\phi(X) = 1) > 0$. Write $\mathcal{A}:=\{\phi:\Prob(\phi(X)=1)>0\}$.

\begin{remark}[Completeness of the information set]
\label{rem:X_complete}
For the impossibility results to bind, $X$ must represent the full rating-time information, including analyst judgment and committee overrides. If the actual process uses an enriched signal $X' = (X, U)$, then $\esssup\, dP_1^X/dP_0^X \leq \esssup\, dP_1^{X'}/dP_0^{X'}$, and the formal impossibility applies only to the coarser signal.
\end{remark}

\begin{remark}[Event and horizon dependence]
\label{rem:event_horizon}
The framework is exact only after the certified event and evaluation horizon have been fixed. Cross-application comparisons are like-for-like only when the same event/horizon pair is imposed throughout.
\end{remark}

\subsection{The precision bound}

\begin{definition}[Discrimination ratio]
\label{def:discrimination}
For a decision rule $\phi$ with $F(\phi) > 0$, the \emph{discrimination ratio} is $\Lambda(\phi) := S(\phi)/F(\phi)$.
\end{definition}

\begin{theorem}[Precision Bound]
\label{thm:precision_bound}
For any decision rule $\phi$ with $F(\phi) > 0$:
\begin{equation}
\label{eq:ppv_formula}
\PPV(\phi) = \frac{\pi \Lambda(\phi)}{\pi \Lambda(\phi) + (1-\pi)}.
\end{equation}
Achieving $\PPV(\phi) \geq \tgt$ requires
\begin{equation}
\label{eq:lambda_req}
\Lambda(\phi) \geq \Lreq(\tgt, \pi) := \frac{\tgt}{1-\tgt} \cdot \frac{1-\pi}{\pi}.
\end{equation}
\end{theorem}

\begin{proof}
Bayes' theorem gives $\PPV(\phi) = S\pi / (S\pi + F(1-\pi))$. Dividing numerator and denominator by $F$ yields~\eqref{eq:ppv_formula}. Solving $\PPV \geq \tgt$ for $\Lambda$ gives~\eqref{eq:lambda_req}.
\end{proof}

Table~\ref{tab:lambda_req} displays the required discrimination for $\tgt = 0.9999$. At $\pi = 0.50$, the requirement is approximately $10{,}000$. The relationship between PPV, base rate, and discrimination is familiar from medical screening \citep{casscells1978, gigerenzer2002}. The contribution here is the calibration to structured finance and the argument that available discrimination fell dramatically short.

\begin{table}[H]
\centering
\caption{Required discrimination $\Lreq$ for $\tgt = 0.9999$ at various base rates.}
\label{tab:lambda_req}
\begin{tabular}{ccc}
\toprule
\textbf{Base rate $\pi$} & \textbf{$(1-\pi)/\pi$} & \textbf{Required $\Lreq$} \\
\midrule
0.90 & 0.11 & 1,100 \\
0.70 & 0.43 & 4,300 \\
0.50 & 1.00 & 10,000 \\
0.30 & 2.33 & 23,300 \\
0.10 & 9.00 & 90,000 \\
\bottomrule
\end{tabular}
\end{table}

\begin{remark}[Interpretation of the four-nines benchmark]
\label{rem:four_nines}
The calibration $\tgt = 0.9999$ is an upper bound on what AAA might mean. Moody's cumulative default rate data imply that the historical corporate AAA five-year default rate corresponds roughly to $\tgt \approx 0.999$ \citep{moodys2006}. Because $\Lreq$ scales like $\tgt/(1-\tgt)$, moving from $0.999$ to $0.9999$ multiplies required discrimination by 10. The main conclusion for pre-crisis CDOs is robust to this change; see Section~\ref{sec:empirical}.
\end{remark}

\subsection{The discrimination ceiling}

Theorem~\ref{thm:precision_bound} establishes how much discrimination is required. The next step is to bound how much is available.

Let $\Prob_0$ and $\Prob_1$ denote the conditional laws of $X$ given $E = 0$ and $E = 1$. We assume the overlap condition $\Prob_1 \ll \Prob_0$ (absolute continuity), so that the likelihood ratio $L(x) := d\Prob_1/d\Prob_0(x)$ is well defined.

\begin{remark}[Plausibility of the overlap condition]
\label{rem:overlap_plausible}
The overlap condition requires that no measurable combination of rating-time observables is literally impossible for instruments that subsequently fail. In structured finance, rating-time information consists of continuous or near-continuous variables, and failure is driven by latent variables not fully identified by observables. The condition is therefore natural. It would fail only if a signal at rating time perfectly ruled out subsequent failure.
\end{remark}

\begin{theorem}[Discrimination Ceiling]
\label{thm:discrimination_ceiling}
For any decision rule $\phi$ with $F(\phi) > 0$,
\begin{equation}
\Lambda(\phi) \leq \esssup L,
\end{equation}
where $\esssup L := \inf\{ c \in [0,\infty] : \Prob_0(L(X) > c) = 0 \}$. The bound is tight: for every $\varepsilon > 0$ (if $\esssup L < \infty$) or every $M < \infty$ (if $\esssup L = \infty$), there exists an admissible rule achieving discrimination within $\varepsilon$ of $\esssup L$ or exceeding $M$, respectively.
\end{theorem}

\begin{proof}
Let $A_\phi := \{x : \phi(x)=1\}$. By the Radon--Nikodym theorem,
\[
\Lambda(\phi)
= \frac{\E_0[\phi(X)L(X)]}{\E_0[\phi(X)]},
\]
which is the $\Prob_0$-average of $L$ over $A_\phi$, hence bounded by $\esssup L$.

For tightness when $\esssup L < \infty$, define $\phi_\varepsilon(x) = \ind\{L(x) \geq \esssup L - \varepsilon\}$. The definition of essential supremum ensures $F(\phi_\varepsilon) > 0$, and on $A_{\phi_\varepsilon}$ one has $L \geq \esssup L - \varepsilon$, giving $\Lambda(\phi_\varepsilon) \geq \esssup L - \varepsilon$.

When $\esssup L = \infty$, define $\phi_M(x) = \ind\{L(x) \geq M\}$. Then $F(\phi_M) > 0$ and $\Lambda(\phi_M) \geq M$.
\end{proof}

The result is closely related to the Neyman--Pearson lemma \citep{neyman1933}. The ceiling $\Lavail := \esssup L$ is a property of the data-generating process, not of any algorithm. The same bound applies to randomized rules, since $\Lambda(\alpha) = \E_0[\alpha(X)L(X)]/\E_0[\alpha(X)] \leq \esssup L$ for any measurable $\alpha: \mathcal{X} \to [0,1]$.

\subsection{The tension ratio and impossibility}

\begin{definition}[Tension ratio]
\label{def:tension}
The \emph{tension ratio} is
\[
\Psi :=
\begin{cases}
\Lreq / \Lavail, & \text{if } \Lavail < \infty,\\[0.4em]
0, & \text{if } \Lavail = \infty.
\end{cases}
\]
\end{definition}

\begin{theorem}[Impossibility]
\label{thm:impossibility}
If $\Psi > 1$, then necessarily $\Lavail < \infty$, and no admissible decision rule can achieve target precision $\tgt$:
\[
\sup_{\phi \in \mathcal{A}} \PPV(\phi)
= \frac{\pi \Lavail}{\pi \Lavail + (1-\pi)}
< \tgt.
\]
\end{theorem}

\begin{proof}
Under $\Prob_1 \ll \Prob_0$, any admissible rule must satisfy $F(\phi) > 0$: if $F(\phi) = 0$, then $\Prob_0(A_\phi) = 0$, and absolute continuity forces $\Prob_1(A_\phi) = 0$, contradicting admissibility.

Every admissible rule therefore satisfies $\Lambda(\phi) \leq \Lavail$ by Theorem~\ref{thm:discrimination_ceiling}. Substituting into~\eqref{eq:ppv_formula}:
\[
\PPV(\phi) \leq \frac{\pi \Lavail}{\pi \Lavail + (1-\pi)}.
\]
Tightness of the discrimination ceiling establishes the supremum. If $\Psi > 1$, then $\Lavail < \Lreq$, and strict monotonicity of~\eqref{eq:ppv_formula} in $\Lambda$ gives the strict inequality.
\end{proof}

\subsection{Coverage-constrained impossibility}

The unconstrained ceiling $\Lavail$ can be driven by extremely small acceptance regions. Because pre-crisis AAA ratings were issued at scale, a refinement is useful.

\begin{definition}[Coverage-constrained ceiling]
\label{def:coverage_available}
Fix $\pi \in (0,1)$ and a minimum issuance rate $q \in (0,1)$. Define the unconditional issuance probability $I_\pi(\phi):=\pi S(\phi)+(1-\pi)F(\phi)$. The \emph{coverage-constrained discrimination ceiling} is
\[
\Lavail^{(q)}(\pi)
:= \sup\left\{\Lambda(\phi) : \phi \in \mathcal{A},\; I_\pi(\phi) \geq q\right\}.
\]
Since the supremum is over a smaller class, $\Lavail^{(q)}(\pi) \leq \Lavail$.
\end{definition}

\begin{proposition}[Coverage-Constrained Impossibility]
\label{prop:coverage_impossibility}
If $\Lreq(\tgt, \pi) > \Lavail^{(q)}(\pi)$, then no decision rule with issuance rate at least $q$ can achieve $\PPV \geq \tgt$.
\end{proposition}

\begin{proof}
Any rule with $I_\pi(\phi) \geq q > 0$ is admissible and under $\Prob_1 \ll \Prob_0$ has $F(\phi) > 0$. The result follows from the definition of $\Lavail^{(q)}(\pi)$ and monotonicity of~\eqref{eq:ppv_formula}.
\end{proof}

The unconstrained tension ratios reported below are therefore conservative: at-scale ratios can only be larger.

\subsection{Structural barriers}

Four features of the structured-finance setting prevent $\Lavail$ from reaching $\Lreq$.

\paragraph{Fixed-information limitation.} The discrimination ceiling is determined by the overlap between $\Prob_0$ and $\Prob_1$. No algorithm can extract discrimination absent from the data.

\paragraph{Regime uncertainty.} The pre-crisis historical sample was drawn overwhelmingly from benign conditions. Enlarging it sharpens within-regime inference without necessarily identifying the dependence structure relevant under stress.

\paragraph{Bottleneck uncertainty.} Weak identification in one parameter block can materially limit overall discrimination. For pre-crisis CDOs, at least three inputs were severely uncertain: default correlation \citep{coval2009}, subprime default probabilities \citep{demyanyk2011}, and recovery rates.

\paragraph{Structural transformation bound.} Deterministic transformations cannot increase the ceiling.

\begin{proposition}[Structural Transformation Bound]
\label{prop:engineering}
Let $Y = g(X)$ be any deterministic measurable transformation. Let $P_i^X$ denote the law of $X$ under $E = i$ and $P_i^Y$ the pushforward. Assume $P_1^X \ll P_0^X$. Then
\[
L_Y(Y)=\E_{P_0^X}[L_X(X)\mid Y]
\qquad P_0^X\text{-a.s.},
\]
and consequently $\esssup L_Y \leq \esssup L_X$.
\end{proposition}

\begin{proof}
The pushforward measures satisfy $P_1^Y \ll P_0^Y$. For any measurable $B$,
\[
P_1^Y(B) = \E_{P_0^X}[\ind_{\{Y \in B\}} L_X(X)] = \E_{P_0^X}[\ind_{\{Y \in B\}} \E_{P_0^X}[L_X(X) \mid Y]].
\]
By the Radon--Nikodym property, $L_Y(Y) = \E_{P_0^X}[L_X(X) \mid Y]$ a.s., so $L_Y(Y) \leq \esssup L_X$ a.s.
\end{proof}

Pooling, tranching, and waterfall rules operate on the same underlying information and cannot create discrimination the data lack. The point connects to \citet{demarzo2005}: pooling and tranching create informationally insensitive securities, a property that runs in the wrong direction for tail certification.

\begin{remark}[Tranching can change $\Lavail$ by changing $E$]
\label{rem:tranching_changes_E}
Proposition~\ref{prop:engineering} holds for a fixed certified event. Changing the tranche structure redefines $E$, which changes the conditional laws and hence $\Lavail$. Deepening subordination can raise the tranche-level ceiling without new information by moving the certification boundary into a region where the signal laws separate more sharply. The empirical tension ratios below are conditional on the structures actually issued.
\end{remark}

\section{Calibration to Pre-Crisis CDOs}
\label{sec:empirical}

\subsection{The informational burden}

Throughout this section, the reference class is the population of senior CDO tranches presented for AAA rating during 2005--2007. The certified event is survival without downgrade to speculative grade over a five-year horizon. The base rate $\pi$ is the fraction of that population that would have survived under this definition. Approximately 83\% of Moody's Aaa-rated mortgage-backed security tranches from 2006 were subsequently downgraded \citep{fcic2011}, and failure rates for ABS CDO tranches were higher still, with more than 80\% of Aaa CDO bonds eventually downgraded to junk \citep{fcic2011, benmelech2009credit}. Since these were the instruments the rating process selected, the unfiltered candidate pool would have had a survival rate no higher, making $\pi = 0.50$ a generous assumption. Results at $\pi = 0.30$ and $\pi = 0.10$ are also reported.

The required discrimination is not an estimate. It follows directly from Bayes' theorem. At $\pi = 0.50$, $\Lreq \approx 10{,}000$. At $\pi = 0.30$, $\Lreq \approx 23{,}300$. The empirical question is whether any feature of the pre-crisis information environment makes a ceiling of this magnitude plausible.

A direct estimate of $\Lavail$ would require fully specified conditional laws of the rating-time signal. This paper does not attempt that exercise. Instead, it derives the burden exactly and asks whether the historical environment offers any affirmative reason to believe the burden was supportable. Standard moment-matching validation does not resolve the question: Theorem~\ref{thm:moment_insufficiency} (Appendix~\ref{app:moments}) shows that matching finitely many distributional moments leaves extreme-tail behavior underdetermined.

For score-based classifiers, the relevant object at a stringent threshold is the local ratio $S(t)/F(t)$, not AUC. In published credit-risk studies, threshold-specific performance at relevant operating points does not approach $10^4$ \citep{demyanyk2011, khandani2010}. Machine-learning gains are measured in percentage points of AUC, not the orders-of-magnitude improvement in extreme-threshold discrimination needed here.

\begin{example}[Binormal benchmark]
\label{ex:binormal}
Under equal-variance binormality ($X \mid E = i \sim N(\mu_i, 1)$ with separability $d' = \mu_1 - \mu_0$), AUC $= \Phi(d'/\sqrt{2})$ and a threshold rule $\phi_t(x) = \ind\{x \geq t\}$ achieves
\[
\Lambda(t) = \frac{\Phi(d' - t)}{\Phi(-t)}.
\]
Table~\ref{tab:binormal} reports achievable discrimination at stringent false-positive rates for AUC values spanning the published credit-prediction literature.

\begin{table}[H]
\centering
\caption{Achievable discrimination $\Lambda$ under equal-variance binormality at two false-positive rates.}
\label{tab:binormal}
\begin{tabular}{cccc}
\toprule
\textbf{AUC} & \textbf{$d'$} & \textbf{$\Lambda$ at $\FPR = 10^{-3}$} & \textbf{$\Lambda$ at $\FPR = 10^{-4}$} \\
\midrule
0.85 & 1.47 & 52 & 121 \\
0.90 & 1.81 & 101 & 283 \\
0.95 & 2.33 & 222 & 818 \\
0.99 & 3.29 & 579 & 3{,}339 \\
\bottomrule
\end{tabular}
\end{table}

\noindent Credit-risk prediction models in the published literature typically achieve AUCs in the range of $0.80$--$0.90$. Under binormality, models in this range achieve $\Lambda$ between 50 and 300 at false-positive rates of $10^{-3}$ to $10^{-4}$. A model-free bound applies as well: since $S \leq 1$, any rule with $\FPR = F$ satisfies $\Lambda \leq 1/F$, so $\Lambda = 10{,}000$ requires $\FPR \leq 10^{-4}$ with near-perfect sensitivity. The binormal model is not a worst case: distributions with the same AUC can produce higher or lower tail $\Lambda$. Bridging the gap from the low hundreds to $10^4$ would, however, require a score distribution far more informative in the tails than any published credit model has demonstrated.
\end{example}

Three diagnostics point in the same direction. First, published subprime prediction studies provide no evidence of threshold-specific discrimination close to $10^4$ \citep{demyanyk2011, khandani2010}. Second, small changes in default correlation shift senior tranche loss probabilities by factors of 10 to 100 \citep{coval2009}. Third, observable CDO characteristics explained little subsequent performance variation \citep{benmelech2009}. None formally bounds $\Lavail$, but together they establish that no observable feature of the pre-crisis information environment supports a belief that tail discrimination reached even the low hundreds.

The burden of proof is asymmetric. A rating agency issuing a AAA certification makes a precision claim; Bayes' theorem determines the discrimination that claim requires. The question is not whether a skeptic can prove $\Lavail < 10{,}000$, but whether the certifier can provide an affirmative basis for $\Lavail \geq 10{,}000$. The published record offers none.

\subsection{Benchmark tension ratios}

The tables below use $\Lavail = 100$ as a benchmark ceiling. Under the binormal calibration in Example~\ref{ex:binormal}, AUC $= 0.90$ yields $\Lambda \approx 101$ at $\FPR = 10^{-3}$ and $\Lambda \approx 283$ at $\FPR = 10^{-4}$. The benchmark is therefore generous relative to published AUCs of $0.80$--$0.90$. Smaller values are also reported to show sensitivity.

\begin{table}[H]
\centering
\caption{Unconstrained tension ratios for pre-crisis CDOs ($\tgt = 0.9999$). The ceiling $\Lavail$ is a conservative benchmark; see Example~\ref{ex:binormal}.}
\label{tab:calibration}
\begin{tabular}{cccc}
\toprule
\textbf{Base rate $\pi$} & \textbf{$\Lreq$} & \textbf{$\Lavail$} & \textbf{$\Psi$} \\
\midrule
0.50 & 10,000 & 100 & 100 \\
0.50 & 10,000 & 50 & 200 \\
0.50 & 10,000 & 20 & 500 \\
0.30 & 23,300 & 100 & 233 \\
0.30 & 23,300 & 50 & 466 \\
0.10 & 90,000 & 100 & 900 \\
\bottomrule
\end{tabular}
\end{table}

\noindent The conclusion is robust to the choice of ceiling. Even at $\Lavail = 500$, a level of discrimination far beyond any published credit model, $\Psi = 20$ at $\pi = 0.50$ and $\Psi \approx 47$ at $\pi = 0.30$ under four-nines.

\subsection{Robustness to the three-nines target}

Moody's cumulative default rate data imply that the historical corporate AAA five-year default rate corresponds roughly to $\tgt \approx 0.999$ \citep{moodys2006}. At that target and $\pi = 0.50$, $\Lreq = 999$. Under $\Lavail = 100$, $\Psi \approx 10$. At $\pi = 0.30$, $\Lreq \approx 2{,}331$ and $\Psi \approx 23$. Even this less stringent benchmark requires threshold-specific discrimination ratios on the order of $10^3$ at operating points relevant to structured-finance tail risk. The published evidence provides no affirmative support for discrimination of that magnitude.

\begin{corollary}[Benchmark Impossibility for Pre-Crisis CDOs]
\label{cor:cdo_impossibility}
Fix the rating-time information set, reference class, downgrade-based certification event, and evaluation window. If $\pi \leq 0.50$, the overlap condition holds, and $\Lavail \leq 100$, then $\Psi \geq 100$ and no admissible decision rule can achieve $\PPV \geq 0.9999$.
\end{corollary}

\begin{proof}
If $\pi \leq 0.50$, then $\Lreq \geq 9{,}999$. With $\Lavail \leq 100$, $\Psi \geq 99.99 > 1$. Theorem~\ref{thm:impossibility} applies.
\end{proof}

\noindent These tension ratios use the unconstrained ceiling. Under the coverage constraint (Definition~\ref{def:coverage_available}), pre-crisis AAA issuance rates of $q \approx 0.05$--$0.10$ would lower the effective ceiling to $\Lavail^{(q)} < \Lavail$, making the reported $\Psi$ conservative.

\subsection{Rescue conditions}

Solving $\Lreq \leq \Lavail$ for $\pi$ gives the minimum base rate for feasibility:
\begin{equation}
\label{eq:rescue}
\pi \geq \frac{\tgt}{\tgt + (1-\tgt)\Lavail}.
\end{equation}
With $\Lavail = 100$ and $\tgt = 0.9999$, this requires $\pi \geq 0.99$.

\begin{table}[H]
\centering
\caption{Minimum base rate for unconstrained feasibility at $\Lavail = 100$.}
\label{tab:rescue}
\begin{tabular}{cccc}
\toprule
\textbf{Target $\tgt$} & \textbf{Min $\pi$} & \textbf{$\Lreq$ at $\pi = 0.50$} & \textbf{$\Lreq$ at $\pi = 0.30$} \\
\midrule
0.99 & 0.50 & 99 & 231 \\
0.995 & 0.67 & 199 & 464 \\
0.999 & 0.91 & 999 & 2,331 \\
0.9999 & 0.99 & 10,000 & 23,300 \\
\bottomrule
\end{tabular}
\end{table}

If AAA means something as weak as 99\%, feasibility depends on the base rate. At the corporate-bond benchmark of 99.9\% or above, the gap is large at any plausible structured-finance base rate. Rescuing the precision claim requires either an overwhelmingly successful candidate pool or a very high discrimination ceiling. Neither condition, let alone both, finds support in the evidence.

\subsection{Revealed performance}

The FCIC estimates that approximately 83\% of mortgage-backed security tranches rated Aaa by Moody's in 2006 were subsequently downgraded; for ABS CDO tranches from the 2006--2007 vintages, more than 80\% of Aaa bonds were eventually downgraded to junk \citep{fcic2011, benmelech2009credit}. Taking a failure rate of $f = 0.90$ as an approximate benchmark for the worst-performing subprime ABS CDO segment:

\begin{definition}[Achieved discrimination]
\label{def:lach}
If the observed predictive value is $\PPV_{\mathrm{obs}} = 1 - f$, the achieved discrimination is
\[
\Lach := \frac{1-f}{f} \cdot \frac{1-\pi}{\pi}.
\]
\end{definition}

At $\pi = 0.50$, $\Lach \approx 0.11$, below 1: the AAA label was negatively informative. This retrospective quantity should not be confused with the ex ante ceiling $\Lavail$. Low achieved discrimination is consistent with low available discrimination but can also reflect misspecification or poor implementation. Its role here is evidentiary.

\begin{proposition}[Prior-Free Discrimination Deficit]
\label{prop:pi_free}
Assume $\PPV_{\mathrm{obs}} \in (0,1)$. The ratio of required to achieved discrimination satisfies
\[
\frac{\Lreq}{\Lach} = \frac{\tgt(1 - \PPV_{\mathrm{obs}})}{(1 - \tgt)\PPV_{\mathrm{obs}}},
\]
which does not depend on $\pi$.
\end{proposition}

\begin{proof}
Both $\Lreq$ and $\Lach$ contain the factor $(1-\pi)/\pi$, which cancels.
\end{proof}

With $\PPV_{\mathrm{obs}} \approx 0.10$ and $\tgt = 0.9999$, this ratio is approximately 90{,}000. At $\tgt = 0.999$, it is approximately 9{,}000.

\begin{remark}[Scope of the prior-free comparison]
\label{rem:pi_free_scope}
The 90{,}000 figure is exact within the fixed retrospective application in which the observed AAA tranches constitute the positively rated instruments for the stated event, horizon, and reference class. If the relevant candidate pool is broader, the figure is subset-specific.
\end{remark}

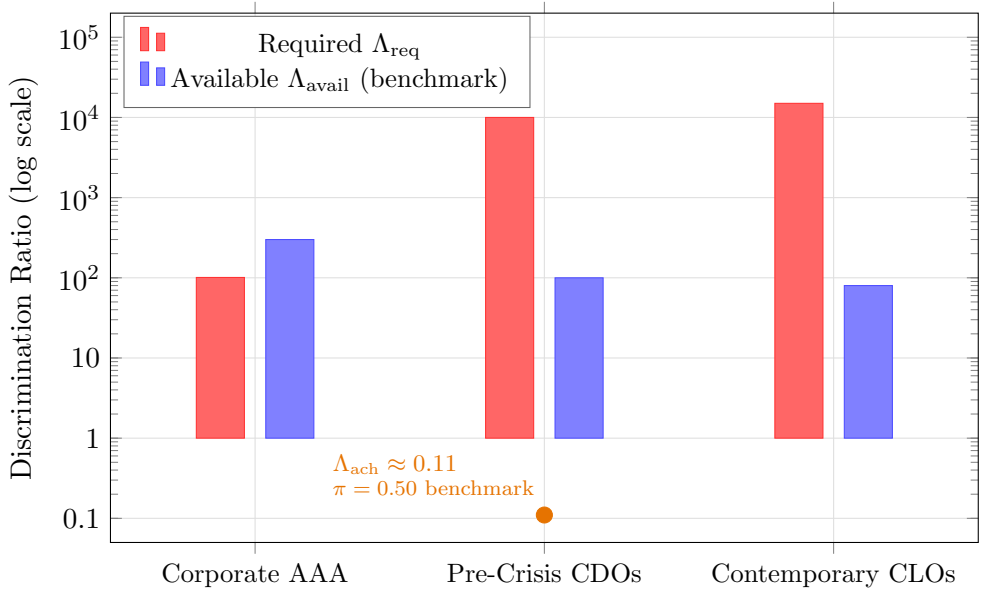
\begin{figure}[H]
\centering
\begin{tikzpicture}
\begin{axis}[
    width=13cm,
    height=8.6cm,
    ymode=log,
    ymin=0.05, ymax=200000,
    ylabel={Discrimination Ratio (log scale)},
    symbolic x coords={Corporate AAA, Pre-Crisis CDOs, Contemporary CLOs},
    xtick=data,
    xticklabel style={font=\small},
    ytick={0.1, 1, 10, 100, 1000, 10000, 100000},
    yticklabels={0.1, 1, 10, $10^2$, $10^3$, $10^4$, $10^5$},
    grid=major,
    grid style={gray!25},
    legend style={
        at={(0.015,0.995)},
        anchor=north west,
        font=\small,
        draw=black!60,
        fill=white,
        inner xsep=6pt,
        inner ysep=4pt
    },
    bar width=18pt,
    ybar=8pt,
    enlarge x limits=0.25,
    clip=false
]

\addplot[fill=red!60, draw=red!80] coordinates {
    (Corporate AAA, 101)
    (Pre-Crisis CDOs, 10000)
    (Contemporary CLOs, 15000)
};
\addlegendentry{Required $\Lreq$}

\addplot[fill=blue!50, draw=blue!70] coordinates {
    (Corporate AAA, 300)
    (Pre-Crisis CDOs, 100)
    (Contemporary CLOs, 80)
};
\addlegendentry{Available $\Lavail$ (benchmark)}

\addplot[
    only marks,
    mark=*,
    mark size=3pt,
    orange!90!black
] coordinates {(Pre-Crisis CDOs, 0.11)};

\node[
    font=\footnotesize,
    orange!90!black,
    align=left,
    anchor=east
] at (axis cs:Pre-Crisis CDOs, 0.35) {$\Lach \approx 0.11$\\[-2pt]\scriptsize $\pi=0.50$ benchmark};

\end{axis}
\end{tikzpicture}
\caption{Discrimination ratios across asset classes (log scale). Red bars: required discrimination under benchmark base-rate assumptions. Blue bars: benchmark ceilings (see Example~\ref{ex:binormal}). Orange marker: achieved discrimination for pre-crisis CDOs ($\pi = 0.50$); $\Lach < 1$ indicates the rating was negatively informative. Cross-asset comparisons are stylized: different applications use different certified events, horizons, and reference classes.}
\label{fig:discrimination_gap}
\end{figure}

\section{Plausibility Check and Contemporary Application}
\label{sec:results}

\subsection{Corporate bonds}

A useful framework should produce feasible calibrations where AAA ratings have historically performed well. For a corporate AAA-candidate reference class, take $\pi \approx 0.99$ and $\tgt = 0.9999$. Then $\Lreq \approx 101$. If available discrimination is in the low hundreds, $\Psi < 1$ and the target is attainable. These are benchmark values for a plausibility check, not event-matched estimates of corporate $\Lavail$ or $\pi$.

The differences are structural. Self-selection produces high base rates among firms seeking AAA. Creditworthiness is directly observable rather than mediated through a correlation model. The historical record spans multiple stress regimes. Defaults are relatively weakly correlated.

\begin{table}[H]
\centering
\caption{Illustrative tension-ratio determinants across asset classes. Cross-column comparisons are stylized benchmarks, not same-contract estimates.}
\label{tab:comparison}
\begin{tabular}{lccc}
\toprule
& \textbf{Corporate AAA} & \textbf{2006--07 CDOs} & \textbf{Contemp.\ CLOs} \\
\midrule
Base rate $\pi$ & $\sim 0.99$ & $\sim 0.30$--$0.50$ & $\sim 0.40$ \\
Required $\Lreq$ & $\sim 100$ & $\sim 10{,}000$--$23{,}000$ & $\sim 15{,}000$ \\
Available $\Lavail$ & $\sim 300$+ & $\sim 50$--$100$ & $\sim 50$--$100$ \\
Tension ratio $\Psi$ & $\sim 0.3$ & $\sim 100$--$470$ & $\sim 150$--$300$ \\
\bottomrule
\end{tabular}
\end{table}

\subsection{Contemporary CLOs}

No AAA-rated CLO tranche has defaulted through principal loss \citep{sfa2020, sp2024clo}. A finite zero-failure record, however, does not establish $\PPV = 1$. The tension ratio framework asks whether rating-time data could support discrimination on the order of $10^4$.

For $\pi = 0.40$ and $\tgt = 0.9999$, $\Lreq \approx 15{,}000$. With $\Lavail = 80$, $\Psi \approx 188$. This section is illustrative, not a theorem-level impossibility proof. The evidence below draws on several distinct events (principal loss, coverage-ratio deterioration, model-test failure, covenant erosion), each bearing on the discrimination question from a different angle. Four features bear on the calibration.

\paragraph{Limited stress data.} The CLO~2.0 era (2010--present) has produced zero AAA defaults under broadly favorable conditions. During 2008--2009, AAA coverage ratios declined sharply, falling from approximately 1.16 to 0.85, meaning that even the most senior rated tranches were undercollateralized in early 2009 \citep{cordell2023}. Survival reflected recovery and the closed-end CLO structure, not ex ante discrimination.

\paragraph{Model failures in real time.} During the 2020 COVID episode, S\&P's model-based test results reported in CLO trustee reports indicated that AAA-rated tranches in 12.5\% of reporting CLOs could not withstand the scenario default rate that a AAA tranche should survive \citep{griffin2023}. The agencies incorporated qualitative overrides rather than downgrade, suggesting that implemented models were not generating the discrimination a stringent target requires.

\paragraph{Sensitivity to correlation.} Typical BSL CLO AAA subordination of 35--38\% requires a large fraction of underlying loans to default before principal loss under standard recovery assumptions \citep{blackrock2025, naic2025}. At historical average default rates and low correlation, this is very unlikely. Under sustained stress with elevated correlation, the buffer could erode materially. The key parameter (stress-regime default correlation) is not identified by benign-period data.

\paragraph{Covenant erosion.} Approximately 90\% of outstanding U.S.\ leveraged loans now lack maintenance covenants, compared with approximately 20\% in 2007 \citep{naic2025}. The historical track record was compiled under collateral conditions that no longer prevail.

By Proposition~\ref{prop:engineering}, once the information set includes collateral data and static deal terms, no derived signal can raise $\Lavail$ beyond the ceiling they imply. Active management incorporating genuinely new post-rating information falls outside the proposition and requires a dynamic extension.

\begin{table}[H]
\centering
\caption{CLO stress indicators. Rows are descriptive; they do not share a single certified event or reference class.}
\label{tab:clo_stress}
\small
\begin{tabular}{p{4.5cm} p{3.5cm} p{5cm}}
\toprule
\textbf{Metric} & \textbf{Value} & \textbf{Source} \\
\midrule
AAA CLO defaults (historical) & 0 observed & \citet{sfa2020} \\
AAA coverage ratio trough (2009) & 0.85 (from 1.16) & \citet{cordell2023} \\
AAA tranches failing S\&P test (June 2020) & 12.5\% & \citet{griffin2023} \\
Typical AAA subordination (BSL CLO) & 35--38\% & \citet{blackrock2025}; \citet{naic2025} \\
Covenant-lite share of outstanding loans & $\sim$90\% (2024) & \citet{naic2025} \\
\bottomrule
\end{tabular}
\end{table}

\subsection{Spread evidence}

If AAA CLOs and AAA corporates carried comparable risk, their spreads should be broadly similar. In practice, AAA CLO spreads persistently exceed AAA corporate spreads by a substantial margin. As of January 2026, the Palmer Square AAA CLO index carried a discount margin of 98 basis points over SOFR \citep{palmersquare2026}, while the ICE BofA AAA US Corporate Index option-adjusted spread stood at approximately 34 basis points \citep{ice2026}. Earlier in the sample, the gap was wider: at the beginning of 2024, new-issue AAA CLO spreads were approximately 160 basis points over SOFR \citep{naic2025} versus a corporate AAA OAS of 38 basis points. Over the ten-year period from January 2014 through December 2024, the average corporate AAA OAS was 62 basis points. Table~\ref{tab:spreads} summarizes the comparison.

\begin{table}[H]
\centering
\caption{AAA CLO versus AAA corporate spreads (basis points). CLO spreads are discount margins over SOFR; corporate spreads are option-adjusted spreads over Treasuries. The conventions differ, so the differential is not a pure credit-risk measure.}
\label{tab:spreads}
\small
\begin{tabular}{lccc}
\toprule
\textbf{Period} & \textbf{AAA CLO} & \textbf{AAA Corporate} & \textbf{Differential} \\
\midrule
Early 2024 & $\sim$160 & 38 & $\sim$122 \\
Year-end 2024 & $\sim$125 & 33 & $\sim$92 \\
January 2026 & 98 & 34 & 64 \\
10-year average & $\sim$130 & 62 & $\sim$68 \\
\bottomrule
\end{tabular}
\end{table}

\noindent The spread conventions differ (discount margin over SOFR versus option-adjusted spread over Treasuries), so the gap is not a pure credit-risk differential. Liquidity, complexity, and investor-base effects all contribute. A persistent differential of 60--120 basis points is nonetheless inconsistent with the view that the two AAA labels represent the same tail safety.

\section{Discussion}
\label{sec:discussion}

\subsection{Limits of moment-based validation}

Pre-crisis validation checked that models matched historical moments. Theorem~\ref{thm:moment_insufficiency} in Appendix~\ref{app:moments} shows that such exercises cannot constrain the extreme-tail behavior relevant to AAA certification. Validation based on finitely many distributional summaries leaves extreme-tail behavior underdetermined.

\subsection{What ratings can defensibly claim}

A calibrated conditional claim requires specifying the target reliability, certified event, evaluation horizon, reference class, and maintained regime assumptions. Such claims differ fundamentally from unconditional tail certification.

Rather than asserting a single probability, a rating report could disclose sensitivity: at historical mean correlation, 99.99\% survival; at historical maximum, 99.50\%; at moderate stress, 92\%; at severe stress, 65\%. This would convey what an unconditional AAA label obscures.

\subsection{A proposal for tension ratio disclosure}

Regulators should require disclosure of the tension ratio for structured products receiving investment-grade ratings: the target $\tgt$, certified event and horizon, reference class, estimated $\pi$ (with sensitivity bounds), $\Lavail$ (specifying whether unconstrained or coverage-constrained), and the resulting $\Psi$.

One may object that $\Lavail$ is difficult to estimate. When the gap spans orders of magnitude, precision is unnecessary. A more serious concern is market destabilization, but the alternative (issuing ratings whose precision claims are unsupported) stores up larger risks.

\section{Conclusion}
\label{sec:conclusion}

The formal contribution is a conditional impossibility theorem: if the discrimination ceiling falls below the Bayes-implied requirement, no admissible decision rule can achieve the target precision. The empirical contribution is narrower. Required discrimination was extraordinarily large. The available public evidence provides no affirmative basis for believing the pre-crisis information environment could support discrimination close to that level, and even generous benchmark ceilings imply tension ratios in the hundreds. The achieved-versus-required gap spans roughly five orders of magnitude under four-nines and four orders under three-nines.

The framework is not uniformly pessimistic. It accommodates corporate AAA ratings, where high base rates and rich information produce low tension ratios. For contemporary CLOs, illustrative calibrations yield tension ratios well above one, and market pricing points in the same direction.

These failures cannot be understood solely as modeling mistakes or incentive failures. The precision target very likely exceeded what the available information could support.

Several extensions are natural. The framework applies to other certification contexts, including catastrophe bonds and sovereign risk. The transformation bound extends to dynamic settings with active management. Empirical work could estimate $\Lavail$ using granular loan-level data spanning both normal and stress outcomes. The interaction between informational impossibility and incentive distortions warrants formal modeling: when the target is infeasible, are incentive problems amplified or attenuated?

Some precision claims exceed what the available information can support. Sound risk management begins by recognizing that limit.

\section*{Declarations of Interest}

The authors report no conflicts of interest. The authors alone are responsible for the content and writing of the paper.

\section*{Use of Artificial Intelligence}

The authors used AI-assisted tools (Overleaf, Claude) to assist with grammar, fluency, and LaTeX formatting. No AI system contributed to the formulation of theoretical arguments, proofs, interpretations, or empirical analyses.

\appendix
\section{Moment Insufficiency}
\label{app:moments}

\begin{theorem}[Moment Insufficiency]
\label{thm:moment_insufficiency}
Let $P$ be a random variable on $[0,1]$ with first $M$ moments $\mu_1, \ldots, \mu_M$, where $\mu_1 < 1$. For every $\varepsilon > 0$, there exist distributions $D_1$ and $D_2$ on $[0,1]$ with
\[
\max_{m \leq M} |\E_{D_i}[P^m] - \mu_m| \leq \varepsilon \quad \textup{for } i = 1, 2,
\]
yet $\E_{D_2}[P^r]/\E_{D_1}[P^r] \to \infty$ as $r \to \infty$.
\end{theorem}

\begin{proof}
We may assume $\varepsilon \in (0,1)$.

\emph{Case 1: $\mu_1 = 0$.} Then $\mu_m = 0$ for all $m \geq 1$. Choose $0 < a < b < 1$ with $a \leq \varepsilon/2$ and $\eta \in (0, \varepsilon/2)$. Let $D_1 = \delta_a$ and $D_2 = (1-\eta)\delta_a + \eta\delta_b$. For every $m \leq M$, $\E_{D_1}[P^m] = a^m \leq \varepsilon/2$ and $\E_{D_2}[P^m] \leq a + \eta \leq \varepsilon$. Moreover,
\[
\frac{\E_{D_2}[P^r]}{\E_{D_1}[P^r]} = (1-\eta) + \eta\left(\frac{b}{a}\right)^r \to \infty.
\]

\emph{Case 2: $\mu_1 \in (0,1)$.} Let $Q$ be a distribution on $[0,1]$ with moments $\mu_1, \ldots, \mu_M$. Fix $\delta = \varepsilon/(4M)$ and $\eta = \varepsilon/2$. Define $D_1$ as the law of $(1-2\delta)U$ where $U \sim Q$, and $D_2 = (1-\eta)D_1 + \eta\delta_{1-\delta}$.

For each $m \leq M$, $|\E_{D_1}[P^m] - \mu_m| = \mu_m(1-(1-2\delta)^m) \leq 2M\delta = \varepsilon/2$, and $|\E_{D_2}[P^m] - \E_{D_1}[P^m]| \leq \eta = \varepsilon/2$. So both approximate the target moments within $\varepsilon$.

Since $\mu_1 > 0$, $\E_Q[U^r] > 0$ for all finite $r$. We have $\E_{D_1}[P^r] = (1-2\delta)^r \E_Q[U^r] \leq (1-2\delta)^r$ and $\E_{D_2}[P^r] \geq \eta(1-\delta)^r$. Therefore
\[
\frac{\E_{D_2}[P^r]}{\E_{D_1}[P^r]} \geq \eta\left(\frac{1-\delta}{1-2\delta}\right)^r \to \infty,
\]
since $(1-\delta)/(1-2\delta) > 1$.
\end{proof}

\newpage
\bibliographystyle{apalike}

\end{document}